\documentclass[a4paper, 10pt, conference]{ieeeconf}      % Use this line for a4
\IEEEoverridecommandlockouts                              % This command is only
\usepackage{color}
\usepackage{graphicx}
\usepackage{amsmath}
\usepackage{amssymb}
\usepackage{amsfonts}
\usepackage{dsfont}
\usepackage{graphics} % for pdf, bitmapped graphics files

\newtheorem{lemma}{Lemma}[section]

\newtheorem{example}[lemma]{Example}

\newtheorem{proposition}[lemma]{Proposition}

\newcommand{\mR}{\mathbb{R}}

\newcommand{\ith}{i^{\tiny{\text{th}}}}
\newcommand{\jth}{j^{\tiny{\text{th}}}}

\newcommand{\Exp}{\mathrm{Exp}}
\newcommand{\Ln}{\mathrm{Ln}}

\DeclareMathOperator{\rank}{rank}

\title{\LARGE \bf Model-order reduction of biochemical reaction networks}

\author{Shodhan Rao\thanks{Center for Systems Biology, University of Groningen, email: {\tt \small s.rao@umcg.nl}} \and Arjan van der Schaft\thanks{Johann Bernoulli Institute for Mathematics and Computer Science, University of Groningen, e-mail: {\tt\small A.J.van.der.Schaft@rug.nl}} \and Karen van Eunen \thanks{Center for Liver, Digestive and Metabolic Diseases (LDS). Department of Paediatrics, Universitair Medisch Centrum Groningen, e-mail: {\tt\small k.van.eunen@med.umcg.nl}} \and Barbara M. Bakker\thanks{Center for Liver, Digestive and Metabolic Diseases (LDS). Department of Paediatrics, Universitair Medisch Centrum Groningen, e-mail: {\tt\small b.m.bakker@med.umcg.nl}} \and Bayu Jayawardhana\thanks{Discrete Technology and Production Automation, University of Groningen, email: {\tt \small b.jayawardhana@rug.nl}} 
}

\begin{document}

\maketitle
\thispagestyle{empty}
\pagestyle{empty}

%%%%%%%%%%%%%%%%%%%%%%%%%%%%%%%%%%%%%%%%%%%%%%%%%%%%%%%%%%%%%%%%%%%%%%%%%%%%%%%%
\begin{abstract}
In this paper we propose a model-order reduction method for chemical reaction networks governed by general enzyme kinetics, including the mass-action and Michaelis-Menten kinetics. The model-order reduction method is based on the Kron reduction of the weighted Laplacian matrix which describes the graph structure of complexes in the chemical reaction network. We apply our method to a yeast glycolysis model, where the simulation result shows that the transient behaviour of a number of key metabolites of the reduced-order model is in good agreement with those of the full-order model.
\end{abstract}

\section{Motivation and introduction}

The dynamics of a biochemical reaction network can be described by a set of differential equations involving a stoichiometric matrix and reaction rates. The solution to these equations describes the evolution and dynamics of the concentrations of all the metabolites or biochemical species in the network. The stoichiometric matrix, which consists of (positive and negative) integer elements, captures the basic conservation laws of the reactions. The reaction rates are functions of the concentrations and they prescribe the dynamics that take place at each reaction in the network based on the underlying enzyme kinetics. The common type of reaction rates are mass-action kinetics and Michaelis-Menten-type kinetics laws which can include competitive and non-competitive inhibition mechanism. 

Thus a biochemical reaction network $\Sigma$ can be written as
\begin{equation*}
\Sigma \ : \  \left\{ 
\begin{array}{rl} 
\dot{x} & = Sv(x) + S_bv_b \\
y & = Cx 
\end{array}\right.
\end{equation*}
where $x$ and $v_b$ denote the vectors of metabolite concentrations and boundary fluxes, respectively, $y$ is the measured concentrations (or the metabolites of interest), the function $v$ is the reaction rates and $C$ denotes the observation matrix of appropriate dimension. The matrix $S$ and $S_b$ are the stoichiometric matrices of the internal and boundary fluxes, respectively.  

%In many situations, not all of the metabolite concentrations can be measured and hence it is desirable to have a model with only the measurable concentrations. This is where \emph{model reduction} plays a role. 
In this paper, we describe a new model-order reduction method to $\Sigma$ by reducing the number of reactions and metabolites that are involved, such that the dynamics of the measured concentration of the reduced model is close to the full-order one (with the external fluxes). Our model-order reduction method is based on the Kron reduction of the underlying weighted Laplacian describing graph structure of complexes in the chemical reaction networks. A similar technique is also employed in the Kron reduction method for model reduction of resistive electrical networks described in \cite{Kron}. 

In our previous works \cite{Rao, vanderSchaft2012}, we have proposed a framework for describing the dynamics of balanced chemical reaction networks governed by mass-action kinetics where, using the notion of complexes, the dynamics can be written by a complex stoichiometric matrix and a symmetric weighted Laplacian matrix. The symmetric Laplacian matrix corresponds to the graph structure of the reactions and the weights come from the underlying equilibrium constants in each reaction. The model-order reduction method that we proposed in \cite{Rao, vanderSchaft2012} is based on the application of Kron reduction to this Laplacian matrix. In this paper, we extend the method to a general class of biochemical reaction networks. 

%We obtain a reduced model from the original model by imposing the condition that certain metabolites remain at constant concentration. These metabolites are usually the ones that are not measurable or the ones that we consider insignificant from a modelling point of view. Consequently the reduced model has fewer number of variables as compared to the original model, and yet the behaviour of a number of significant metabolites in the reduced network is approximately the same as in the original network. Our model reduction method is useful from a computational point of view, specially when we need to deal with models of huge biochemical reaction networks.     

For biochemical reaction networks, the model-order reduction technique for nonlinear systems is still underdeveloped. The singular perturbation method and quasi steady-state approximation (QSSA) approach are the most commonly used techniques, where the reduced state contains a part of the metabolites of the full model. In the thesis by H\"ardin \cite{Hardin2010}, the QSSA approach is extended by considering the higher-order approximation in the computation of quasi steady-state. Sunn\aa ker {\it et al.} in \cite{Sunnaker2011} proposed a reduction method by identifying variables that can be lumped together and can be used to infer back the original state. In Prescott \& Papachristodoulou \cite{Prescott2012}, a method to compute the upper-bound of the error is proposed based on sum-of-squares programming. The application of these techniques to general kinetics laws, such as Michaelis-Menten, poses a significant computational problem.  

Our proposed method, on the other hand, offers a simple approach to model-order reduction that can be extended directly to general kinetics. Moreover, the resulting reduced-order model retains the structure of the kinetics and gives result to a reduced biochemical reaction network, which enables a direct biochemical interpretation.  
      
We show the application of our model reduction technique to the yeast glycolysis model described in \cite{Karen}. We have simulated the transient behaviour of the metabolites that are not eliminated during the model reduction procedure. We show that there is a good agreement between the transient behaviour of the concentration of most of such metabolites when comparing the full network to the reduced network.

%Our approach relies on a specific compact mathematical formulation describing the dynamics of the concentrations of metabolites in a biochemical reaction network governed by enzyme kinetics. We make use of concepts from graph theory and stochiometry of chemical reactions for the derivation. We then show how the derived formulation can be used in obtaining a reduced order model that has approximately the same behaviour as the original model of a biochemical reaction network. 

The paper is organized as follows. In Section II, we introduce tools from stoichiometry of reactions and graph theory that are required to derive our mathematical formulation. In section III, we explain enzyme kinetics, and then derive our formulation. In section IV, we propose our model reduction method. In section V, we show the application of our method to a yeast glycolysis model and in section VI, we present conclusions based on our results.

\noindent\emph{\bf Notation}:  The space of ${n}$ dimensional real vectors is denoted by $\mathbb{R}^{{n}}$,
and the space of ${m}\times {n}$ real matrices by $\mathbb{R}^{{m}\times {n}}$. The space of ${n}$ dimensional real vectors consisting of all strictly positive entries is denoted by $\mR_+^{n}$ and the space of ${n}$ dimensional real vectors consisting of all nonnegative entries is denoted by $\bar{\mR}_+^{n}$. The rank of a real matrix $A$ is denoted by $\rank A$. $\mathds{1}_m$ denotes a vector of dimension $m$ with all entries equal to 1. The time-derivative $\frac{dx}{dt}(t)$ of a vector $x$ depending on time $t$ will be usually denoted by $\dot{x}$. $I_n$ denotes an identity matrix of dimension $n$.

Define the mapping
$\mathrm{Ln} : \mathbb{R}_+^m \to \mathbb{R}^m, \quad x \mapsto \mathrm{Ln}(x),$
as the mapping whose $i$-th component is given as
$\left(\mathrm{Ln}(x)\right)_i := \mathrm{ln}(x_i).$
Similarly, define the mapping
$\mathrm{Exp} : \mathbb{R}^m \to \mathbb{R}_+^m, \quad x \mapsto \mathrm{Exp}(x),$
as the mapping whose $i$-th component is given as
$\left(\mathrm{Exp}(x)\right)_i := \mathrm{exp}(x_i).$

\section{Chemical reaction network structure}
In this section, we introduce the tools necessary in order to derive our mathematical formulation of the dynamics of chemical reaction networks. First we introduce the concept of stoichiometric matrix of a reaction network. We then introduce the concept of a complex graph, which was first introduced in the work of Horn \& Jackson and Feinberg (\cite{HornJackson, Horn, Feinberg}). 

\subsection{Stoichiometry}\label{sec:stoich}
Consider a chemical reaction network involving $m$ chemical species (metabolites), among which $r$ chemical reactions take place. The basic structure underlying the dynamics of the vector $x \in \mathbb{R}_+^m$ of concentrations $x_i, i=1,\cdots, m,$ of the chemical species is given by the {\it balance laws} $\dot{x} = Sv$, where $S$ is an $m \times r$ matrix, called the {\it stoichiometric matrix}. The elements of the vector $v \in \mathbb{R}^r$ are commonly called the (reaction) {\it fluxes} or {\it rates}. The stoichiometric matrix $S$, which consists of (positive and negative) integer elements, captures the basic conservation laws of the reactions. It already contains useful information about the network dynamics, {\it independent} of the precise form of the reaction rate $v(x)$. Note that the reaction rate depends on the governing law prescribing the dynamics of the reaction network.

We now show how to construct the stoichiometric matrix for a reaction network with the help of an example shown in Fig. \ref{fig:eg1}.
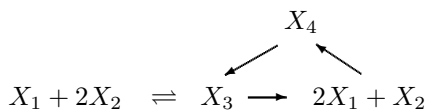
\begin{figure}[h]
\centerline{
  \scalebox{1}{
     %%Created by jPicEdt 1.4.1_03: mixed JPIC-XML/LaTeX format
%%Thu Jul 26 15:47:23 CEST 2012
%%Begin JPIC-XML
%<?xml version="1.0" standalone="yes"?>
%<jpic x-min="14" x-max="54" y-min="0" y-max="10" auto-bounding="true">
%<text text-vert-align= "center-v"
%	 anchor-point= "(14,0)"
%	 fill-style= "none"
%	 text-frame= "noframe"
%	 text-hor-align= "center-h"
%	 >
%$X_1 + 2X_2$
%</text>
%<text text-vert-align= "center-v"
%	 anchor-point= "(34,0)"
%	 fill-style= "none"
%	 text-frame= "noframe"
%	 text-hor-align= "center-h"
%	 >
%$X_3$
%</text>
%<text text-vert-align= "center-v"
%	 anchor-point= "(27,0)"
%	 fill-style= "none"
%	 text-frame= "noframe"
%	 text-hor-align= "center-h"
%	 >
%$\rightleftharpoons$
%</text>
%<multicurve right-arrow= "head"
%	 fill-style= "none"
%	 points= "(38,0);(38,0);(43,0);(43,0)"
%	 />
%<text text-vert-align= "center-v"
%	 anchor-point= "(54,0)"
%	 fill-style= "none"
%	 text-frame= "noframe"
%	 text-hor-align= "center-h"
%	 >
%$2X_1 + X_2$
%</text>
%<multicurve left-arrow= "head"
%	 fill-style= "none"
%	 points= "(35,3);(35,3);(42,7);(42,7)"
%	 />
%<text text-vert-align= "center-v"
%	 anchor-point= "(45,10)"
%	 fill-style= "none"
%	 text-frame= "noframe"
%	 text-hor-align= "center-h"
%	 >
%$X_4$
%</text>
%<multicurve left-arrow= "head"
%	 fill-style= "none"
%	 points= "(47,7);(47,7);(53,3);(53,3)"
%	 />
%</jpic>
%%End JPIC-XML
%LaTeX-picture environment using emulated lines and arcs
%You can rescale the whole picture (to 80% for instance) by using the command \def\JPicScale{0.8}
\ifx\JPicScale\undefined\def\JPicScale{1}\fi
\unitlength \JPicScale mm
\begin{picture}(54,10)(0,0)
\put(14,0){\makebox(0,0)[cc]{$X_1 + 2X_2$}}

\put(34,0){\makebox(0,0)[cc]{$X_3$}}

\put(27,0){\makebox(0,0)[cc]{$\rightleftharpoons$}}

\linethickness{0.3mm}
\put(38,0){\line(1,0){5}}
\put(43,0){\vector(1,0){0.12}}
\put(54,0){\makebox(0,0)[cc]{$2X_1 + X_2$}}

\linethickness{0.3mm}
\multiput(35,3)(0.21,0.12){33}{\line(1,0){0.21}}
\put(35,3){\vector(-2,-1){0.12}}
\put(45,10){\makebox(0,0)[cc]{$X_4$}}

\linethickness{0.3mm}
\multiput(47,7)(0.18,-0.12){33}{\line(1,0){0.18}}
\put(47,7){\vector(-3,2){0.12}}
\end{picture}
     }
   }
\caption{Example of a reaction network}
\label{fig:eg1}
\end{figure}

Note that the above reaction has 5 reactions among four species $X_1$, $X_2$, $X_3$ and $X_4$. Since the stoichiometric matrix maps the space of reactions to the space of species, it has dimension $4 \times 5$. The entry of $S$ corresponding to the $\ith$ row and $\jth$ column is obtained by subtracting the number of moles of the $\ith$ species on the product side from that on the substrate side for the $\jth$ reaction. Thus
\[
S=\begin{bmatrix}
-1 & 1 & 2 & -2 & 0 \\
-2 & 2 & 1 & -1 & 0 \\
1 & -1 & -1 & 0 & -1 \\
0 & 0 & 0 & 1 & 1
\end{bmatrix}
\]
for the reaction network depicted in Fig. \ref{fig:eg1}. 

\subsection{Complex graph}
The structure of a chemical reaction network cannot be directly captured by an ordinary graph. Instead, we will follow an approach originating in the work of Horn and Jackson \cite{HornJackson}, introducing the space of {\it complexes}. The set of complexes of a chemical reaction network is simply defined as the union of all the different left- and righthand sides (substrates and products) of the reactions in the network. Thus, the complexes corresponding to the network (\ref{fig:eg1}) are $X_1+2X_2$, $X_3$, $2X_1+X_2$ and $X_4$.

The expression of the complexes in terms of the chemical species is formalized by an $m \times c$ matrix $Z$, whose $\alpha$-th column captures the expression of the $\alpha$-th complex in the $m$ chemical species. For example, for the network depicted in Figure \ref{fig:eg1},
\[
Z=\begin{bmatrix} 1 & 0 & 2 & 0 \\ 2 & 0 & 1 & 0 \\0 & 1 & 0 & 0 \\ 0 & 0 & 0 & 1\end{bmatrix}.
\]
We will call $Z$ the {\it complex stoichiometric matrix} of the network. Note that by definition all elements of the matrix $Z$ are non-negative integers.

Since the complexes are the left- and righthand sides of the reactions, they can be naturally associated with the vertices of a {\it directed graph} $\mathcal{G}$ with edges corresponding to the reactions. Formally, the reaction $ \alpha  \longrightarrow \beta$ between the $\alpha$-th (reactant) and the $\beta$-th (product) complex defines a directed edge with tail vertex being the $\alpha$-th complex and head vertex being the $\beta$-th complex. The resulting graph will be called the {\it complex graph}.

Recall, see e.g. \cite{Bollobas}, that any graph is defined by its {\it incidence matrix} $B$. This is a $c \times r$ matrix, $c$ being the
number of vertices and $r$ being the number of edges, with $(\alpha,j)$-th
element equal to $-1$ if vertex $\alpha$ is the tail vertex of edge $j$ and $1$ if vertex $\alpha$ is the head vertex of edge $j$, while $0$ otherwise. Thus the structure of the complex graph is described by a $c \times r$ incidence matrix $B$.

Obviously, there is a close relation between the matrix $Z$ and the stoichiometric matrix $S$. In fact, it is easily checked that
\[
S = ZB, \quad \text{hence} \quad \dot{x}=ZBv(x)
\]
with $B$ the incidence matrix of the complex graph.

In many cases of interest, especially in biochemical reaction networks, chemical reaction networks are intrinsically {\it open}, in the sense that there is a continuous exchange with the environment (in particular, inflow and outflow of chemical species and connection to other reaction networks). In this paper, we are particularly interested in inflows to and outflows from the complexes of the network. This will be modeled by a vector $v_b \in \mathbb{R}^c$ consisting of $c$ {\it boundary} (or, exchange) {\it fluxes}, leading to an extended model
\begin{equation}\label{stoichiometry}
\dot{x} = ZBv(x)+ Zv_b
\end{equation}

\section{The dynamics of networks governed by enzyme kinetics}
Recall that for a chemical reaction network, the relation between the reaction rates and species concentrations depends on the governing laws of the reactions involved in the network. In this section, we explain this relation for reaction networks governed by enzyme kinetics. In other words, if $v$ denotes the vector of reaction rates and $x$ denotes the species concentration vector, we show how to construct $v(x)$ for reaction networks governed by enzyme kinetics. 

Let $Z_{\mathcal{S}_j}$ denote the column of the complex stoichiometric matrix $Z$ corresponding to the substrate $\mathcal{S}_j$ of the $j$-th reaction of a chemical reaction network. Let $k_j$ denote the rate constant of the $\jth$ reaction of the network. Then the reaction rate of the $\jth$ reaction of the chemical reaction network between the $\jth$ substrate $\mathcal{S}_j$ and the $\jth$ product $\mathcal{P}_j$ is given by
\begin{equation}\label{eq:EK}
v_j(x)=d_j(x)k_{j} \exp\big(Z_{\mathcal{S}_j}^T \mathrm{Ln}(x)\big),
\end{equation}
where for $j=1,\ldots,r$, $d_j:\mR_+^m \rightarrow \mR_+$ is a rational function of its argument. Note that if the governing law of the $\jth$ reaction is \emph{mass action kinetics}, then $d_j(x)=1$.

As an example, consider the reaction 
\begin{equation}\label{eq:react}
X_1+3X_2 \longrightarrow X_3 +3X_4
\end{equation}
governed by Michaelis Menten kinetics which is the most common form of enzyme kinetics. For $i=1,\ldots,4$, let $x_i$ denote the concentration of the species $X_i$. Define $x:=[x_1 \quad x_2 \quad x_3 \quad x_4]^T$. The matrices $B$, $Z$ and $S$ for the reaction (\ref{eq:react}) are given by
\[
B = \begin{bmatrix}
-1 \\
1
\end{bmatrix} \quad \quad Z = \begin{bmatrix}
1 & 0 \\
3 & 0 \\
0 & 1 \\
0 & 3
\end{bmatrix} \quad \quad
S = \begin{bmatrix}
-1 \\
-3 \\
1 \\
3
\end{bmatrix}
\]
Note that $Z_{\mathcal{S}}$ is given by $Z_{\mathcal{S}}=\begin{bmatrix} 1 & 3 & 0 & 0 \end{bmatrix}^{T}$.
%\[
%Z_{\mathcal{S}} = \begin{bmatrix}
%1 \\
%3 \\
%0 \\
%0
%\end{bmatrix} 
%\]
Let $K_1$, $K_2$, $K_3$ and $K_4$ denote the ``Michaelis" constants of the species $X_1$, $X_2$, $X_3$ and $X_4$ respectively. Let $V_f$ denote the maximum reaction rate of (\ref{eq:react}). In this case, $k=\frac{V_f}{K_1K_2^3}$. The net rate of the reaction (\ref{eq:react}) is given by
\[
v = d(x)kx_1x_2^3 = d(x)k \exp\big(Z_{\mathcal{S}}^T \mathrm{Ln}(x)\big)
\]
where the expression for $d(x)$ depends on the model used to define the rate of the reaction. One possibility for $d(x)$ is
\[
d(x)= \frac{1}{\left(1+\frac{x_1}{K_1}+\frac{3x_2}{K_2}\right)\left(1+\frac{x_3}{K_3}+\frac{3x_4}{K_4}\right)}
\]

Based on the formulation in (\ref{eq:EK}), we can describe the complete reaction network dynamics as follows. Let the reaction rate for the complete set of reactions be given by the vector $v(x)= \begin{bmatrix} v_1(x) & \cdots & v_r(x) \end{bmatrix}^T$. 
For every $\sigma,\pi \in \{1,\cdots, c\}$, define $a_{\pi \sigma}(x):=k_{j}d_j(x)$ if $(\sigma,\pi)=(\mathcal{S}_j,\mathcal{P}_j)$, $j \in \{1,\ldots,r\}$ and $a_{\sigma \pi}:=0$ otherwise. Define the {\it weighted adjacency matrix} $A$ of the complex graph as the matrix with $(\sigma,\pi)$-th element $a_{\sigma \pi}$, where $\sigma,\pi \in \{1, \cdots,c\}$.
Furthermore, define $L(x) := \Delta(x) - A(x)$, where $\Delta$ is the diagonal matrix whose $(\rho,\rho)$-th element is equal to the sum of the elements of the $\rho$-th column of $A$. Hereafter we refer to $L(x)$ as the \emph{weighted Laplacian} of the complex graph associated with the given network with species concentration vector $x$. It is a matter of straightforward verification to check that $\mathds{1}_c^T L(x) = 0$. It can also be verified that the vector $Bv(x)$ for the mass action reaction rate vector $v(x)$ is equal to $-L(x) \Exp \left(Z^T \Ln(x) \right)$, where the mapping $\Exp : \mathbb{R}^c \to \mathbb{R}^c_+$ has been defined at the end of the Introduction.
Hence the dynamics can be compactly written as
\[
\dot{x} = - Z L(x) \mathrm{Exp} \left(Z^T \mathrm{Ln}(x)\right)
\]
A similar expression of the dynamics corresponding to mass action kinetics, in less explicit form, was already obtained in \cite{Sontag}.

\subsection{The linkage classes of a complex graph} \label{sec:link}
A \emph{linkage class} of a chemical reaction network is a maximal set of complexes $\{\mathcal{C}_1,\ldots,\mathcal{C}_k\}$ such that $\mathcal{C}_i$ is connected by reactions to $\mathcal{C}_j$ for every $i,j \in \{1,\ldots,k\}, i \neq j$. It can be easily verified that the number of linkage classes $(\ell)$ of a network, which is equal to the number of connected components of the complex graph corresponding to the network, is given by $\ell =c-$ rank$(B)$ (the number of linkage classes in the terminology of \cite{HornJackson, Feinberg}).

\section{Model reduction}
For many purposes one may wish to reduce the number of dynamical equations of a chemical reaction network in such a way that the behaviour of a number of key metabolites is approximated in a satisfactory way. We propose a novel method for model reduction of chemical reaction networks governed by enzyme kinetics. Our method is based on the Kron reduction method for model reduction of resistive electrical networks described in \cite{Kron}; see also \cite{vdsSCL}. 

\subsection{Description of the method}
Our model reduction method is based on \emph{reduction of the complex graph} associated with the network. It is based on the following result regarding Schur complements of weighted Laplacian matrices.
\begin{proposition}\label{prop:WL}
Consider a chemical reaction network with weighted Laplacian matrix $L(x) \in \mR^{c \times c}$ corresponding to the concentration vector $x$. Let $\mathcal{V}$ denote the set of vertices of the complex graph associated with the network. Then for any subset of vertices $\mathcal{V}_r \subset \mathcal{V}$, the Schur complement $\hat{L}(x)$ of $L(x)$ with respect to the indices corresponding to $\mathcal{V}_r$ satisfies the following properties:
\begin{enumerate}
\item All diagonal elements of $\hat{L}(x)$ are positive and off-diagonal elements are nonnegative for $x \in \mR_+^m$.
\item $\mathds{1}_{\hat{c}}^T\hat{L}(x)=0$, where $\hat{c}:=c-\text{dim}(\mathcal{V}_r)$.
\end{enumerate} 
\end{proposition}
\begin{proof}
(\emph{1}.) Follows from the proof of \cite[Theorem 3.11]{Niezink}; see also \cite{vdsSCL} for the case of symmetric $L$.

\medskip{}

(\emph{2}.) Without loss of generality, assume that the last $c-\hat{c}$ rows and columns of $L(x)$ correspond to the vertex set $\mathcal{V}_r$. Consider a partition of $L(x)$ given by
\begin{equation}\label{eq:part}
L(x)=\begin{bmatrix}
L_{11}(x) & L_{12}(x)\\
L_{21}(x) & L_{22}(x)
\end{bmatrix}
\end{equation}
where $L_{11}(x) \in \mR^{\hat{c}\times \hat{c}}$, $L_{12}(x) \in \mR^{\hat{c} \times (c-\hat{c})}$, $L_{21}(x)\in \mR^{(c-\hat{c})\times \hat{c}}$ and $L_{22}(x)\in \mR^{(c-\hat{c})\times(c-\hat{c})}$. It is easy to see that
\[
\hat{L}(x)=L_{11}(x)-L_{12}(x)L_{22}(x)^{-1}L_{21}(x)
\]
Since $\mathds{1}_{c}^TL(x)=0$, we obtain
\[
\mathds{1}_{\hat{c}}^TL_{11}(x)+\mathds{1}_{c-\hat{c}}^TL_{21}(x) = \mathds{1}_{\hat{c}}^TL_{12}(x)+\mathds{1}_{c-\hat{c}}^TL_{22}(x) = 0
\]
Using the above equations, we get
\begin{eqnarray*}
\mathds{1}_{\hat{c}}^T\hat{L}(x) = \mathds{1}_{\hat{c}}^T\big(L_{11}(x)-L_{12}(x)L_{22}(x)^{-1}L_{21}(x)\big) \\ 
= -\mathds{1}_{c-\hat{c}}^TL_{21}(x)+\mathds{1}_{c-\hat{c}}^TL_{22}(x)L_{22}(x)^{-1}L_{21}(x)=0
\end{eqnarray*}
\end{proof}
From the above result, it follows that $\hat{L}(x)$ obeys all the properties of the weighted Laplacian of a reaction network corresponding to the complex graph with vertex set $\mathcal{V}-\mathcal{V}_r$. Thus Proposition \ref{prop:WL} can be directly applied to the complex graph, yielding a reduction of the chemical reaction network by reducing the number of complexes. Consider a reaction network with boundary fluxes described as in equation (\ref{stoichiometry})
\begin{equation}\label{eq:stanform}
\Sigma: \quad \dot{x} = Z\left(v_b -  L(x) \mathrm{Exp} \big(Z^T \mathrm{Ln}(x)\big)\right)
\end{equation}
Reduction will be performed by {\it deleting certain complexes in the complex graph}, resulting in a reduced complex graph with weighted Laplacian $\hat{L}(x)$. Furthermore, leaving out the corresponding columns of the complex stoichiometric matrix $Z$ one obtains a reduced complex stoichiometric matrix $\hat{Z}$ (with as many columns as the remaining number ($\hat{c}$) of complexes in the complex graph). Consider a corresponding partition of $L$ given by equation (\ref{eq:part}). Define $P:=\begin{bmatrix}I_{\hat{c}} & -L_{12}L_{22}^{-1}\end{bmatrix}$. Now consider the reduced model
\begin{eqnarray}\label{reduced}
\nonumber \hat{\Sigma}: \quad \dot{x} &=& \hat{Z}\left(Pv_b-\hat{L}(x) \Exp \big(\hat{Z}^T \Ln(x)\big)\right) \\
&=& \hat{Z}P\left(v_b -  L(x) \mathrm{Exp} \big(Z^T \mathrm{Ln}(x)\big)\right).
\end{eqnarray}
Note that $\hat{\Sigma}$ is again a {\it chemical reaction network} governed by enzyme kinetics, with a reduced number of complexes and with, in general, a different set of boundary fluxes and reactions (edges of the complex graph).  

An interpretation of the reduced network $\hat{\Sigma}$ can be given as follows. Consider a subset $\mathcal{V}_r$ of the set of all complexes, which we wish to leave out in the reduced network. Consider the partition of $L(x)$ as given by equation (\ref{eq:part}) and corresponding partitions of $Z$ and $v_b$ given by
\[
Z = \begin{bmatrix} Z_1 & Z_2 \end{bmatrix}; \qquad v_b=\begin{bmatrix}v_{b1} & v_{b2}\end{bmatrix}^T,
\]
where $\mathcal{V}_r$ corresponds to the last part of the indices (denoted by $2$), in order to write out the dynamics of $\Sigma$ as
\[
\dot{x} = Z\begin{bmatrix}v_{b1} \\ v_{b2}\end{bmatrix}-Z\begin{bmatrix} L_{11}(x) & L_{12}(x) \\ L_{21}(x) & L_{22}(x) \end{bmatrix} 
\begin{bmatrix} \Exp \big( Z_1^T \Ln (x)\big) \\ \Exp \big(Z_2^T \Ln(x)\big) \end{bmatrix}
\]
Consider now the auxiliary dynamical system
\[
\begin{bmatrix} \dot{y}_1 \\ \dot{y}_2 \end{bmatrix} = \begin{bmatrix}v_{b1} \\ v_{b2}\end{bmatrix} - \begin{bmatrix} L_{11}(x) & L_{12}(x) \\ L_{21}(x) & L_{22}(x) \end{bmatrix} 
\begin{bmatrix} w_1 \\ w_2 \end{bmatrix}
\]
where we impose the constraint $\dot{y}_2 =0$. It follows that 
\[
w_2 = - L_{22}(x)^{-1}(v_{b2}-L_{21}(x)w_1), 
\]
leading to the reduced auxiliary dynamics
\begin{eqnarray*}
\dot{y}_1 &=& (v_{b1}-L_{12}(x)L_{22}(x)^{-1}v_{b2})-\hat{L}(x) w_1 \\
&=& Pv_{b}-\hat{L}(x) w_1
\end{eqnarray*}
Putting back in $w_1 = \Exp \big(\hat{Z}_1^T \Ln(x)\big)$, making use of 
$
\dot{x} = Z_1 \dot{y}_1 + Z_2 \dot{y}_2 = Z_1 \dot{y}_1 = \hat{Z} \dot{y}_1
$, we then obtain the reduced network $\hat{\Sigma}$ given in (\ref{reduced}). An appropriate choice of $\mathcal{V}_r$ will ensure that some of the elements of $x$ have derivative zero in $\hat{\Sigma}$ leading to lesser number of state variables in $\hat{\Sigma}$ as compared to $\Sigma$.

Assume that a given biochemical reaction network is asymptotically stable around an equilibrium. When perturbed from the equilibrium, such a reaction network has certain species reaching their equilibrium much faster than the remaining ones. The principle behind our model reduction method is to impose the condition that complexes entirely made up of such species remain at constant concentrations.
\begin{example}\rm \label{eg:2}
We now consider an example of a simple reversible reaction network governed by Michaelis-Menten kinetics and apply the model reduction procedure described in this paper. This reaction is shown below.
\begin{equation}\label{eq:eg}
X_1+X_2 \rightleftharpoons X_3+X_4 \rightleftharpoons X_5+X_6
\end{equation}
For $i=1,\ldots,6$, let $x_i$ denote the concentration of $X_i$. Let $v_1$ and $v_2$ denote overall reaction rates in the forward direction of the first and the second reversible reactions respectively. For $i=1,\ldots,4$, let $K_{m_1i}$ denote the Michaelis constant of $X_i$ for the first reaction. For $i=3,\ldots,6$, let $K_{m_2i}$ denote the Michaelis constant of $X_i$ for the second reaction. For $i=1,2$, let $k_{i}^f$ and $k_{i}^r$ denote the forward and reverse rate constants of the $\ith$ reaction. Define $x:=$ col$(x_1,x_2,x_3,x_4,x_5,x_6)$,
\begin{eqnarray*}
p_1(x):= \left(1+\frac{x_1}{K_{m_11}}+\frac{x_2}{K_{m_12}}\right)\left(1+\frac{x_3}{K_{m_13}}+\frac{x_4}{K_{m_14}}\right), \\
p_2(x):= \left(1+\frac{x_3}{K_{m_23}}+\frac{x_4}{K_{m_24}}\right)\left(1+\frac{x_5}{K_{m_25}}+\frac{x_6}{K_{m_26}}\right).
\end{eqnarray*}
Then $v_1 = \frac{k_1^fx_1x_2-k_1^rx_3x_4}{p_1(x)}$ and $v_2 = \frac{k_2^fx_3x_4-k_2^rx_5x_6}{p_2(x)}$. Now consider the reduced network
\[
X_1+X_2 \rightleftharpoons X_5+X_6
\]
obtained by deleting the complex $X_3+X_4$ from the network (\ref{eq:eg}). Applying the procedure described in this section, we obtain the following expression for the overall rate ($v$) in the forward direction of the reduced network:
\begin{equation}\label{eq:MM}
v=\frac{k_3^fx_1x_2-k_3^rx_5x_6}{1+\frac{x_1}{K_{m_31}}+\frac{x_2}{K_{m_32}}+\frac{x_5}{K_{m_35}}+\frac{x_6}{K_{m_36}}}.
\end{equation}
where $k_3^f,k_3^r,K_{m_31},K_{m_32},K_{m_35},K_{m_36}$ are positive constants. Note that equation (\ref{eq:MM}) represents the rate of a reaction governed by Michaelis-Menten kinetics. Thus for this example, our procedure yields a reduced model with 6 parameters while the original model has 12 parameters.
\end{example}

\subsection{Effect of Model Reduction}\label{sec:effect}
In this section, we show the effect of our model reduction method on a particular type of networks with a single linkage class. In other words, we give an interpretation of our reduced model in terms of its corresponding full model for a particular type of networks. Note that deletion of a set of complexes in one linkage class does not affect the reactions of the other linkage classes of the network. 
 
\begin{equation}\label{Type1}
\text{Full Network:} \qquad \mathcal{C}_1 \rightleftharpoons \mathcal{C}_2 \rightleftharpoons \mathcal{C}_3 \rightleftharpoons \cdots \cdots \rightleftharpoons \mathcal{C}_n
\end{equation}
\begin{equation}\label{Type1red}
\text{Reduced Network:} \qquad \mathcal{C}_1 \rightleftharpoons \mathcal{C}_3\rightleftharpoons \cdots \cdots \rightleftharpoons \mathcal{C}_n
\end{equation}

Consider a reaction network with reversible reactions occuring between consecutive elements of the set of distinct complexes $\{\mathcal{C}_1,\mathcal{C}_2, \ldots, \mathcal{C}_n\}$ as in (\ref{Type1}). The reduced network obtained by deleting the complex $\mathcal{C}_2$ looks as in (\ref{Type1red}). Instead of the two reversible reactions, $\mathcal{C}_1 \rightleftharpoons \mathcal{C}_2$ and $\mathcal{C}_2 \rightleftharpoons \mathcal{C}_3$ in the full network, there is one reaction $\mathcal{C}_1 \rightleftharpoons \mathcal{C}_3$ in the reduced network. This reaction is a reversible reaction governed by enzyme kinetics. The expression for the rate of this reaction in terms of the expression for the rates of the first two reversible reactions of the full network can be found using the technique described in this section. All the remaining reactions of the reduced network occur in the same way as in the full network. 

Assuming that the network is asymptotically stable around a certain equilibrium, when perturbed from the equilibrium, the transient behaviour of the metabolites involved in the complexes of the reduced model will approximately be the same as that of the full model if the metabolites involved in $\mathcal{C}_2$ reach their steady states much faster than the remaining metabolites. In this case, we can safely impose the condition that the metabolites involved in $\mathcal{C}_2$ are at constant concentration in order to obtain the reduced model (\ref{Type1red}) with similar transient behaviour as that of (\ref{Type1}). The rule of induction may be applied in order to further reduce the model by deleting more complexes. 

A special case of networks (\ref{Type1}) is the following:
\begin{equation}\label{eq:sp}
\mathcal{C}_1 \rightleftharpoons \mathcal{C}_2
\end{equation}
Deletion of the complex $\mathcal{C}_2$ in this case is equivalent to deletion of the linkage class from the network. Such a deletion provides a close approximation to the original network if the reaction (\ref{eq:sp}) has very little effect on the dynamics of the network, i.e if the reaction reaches its steady state much faster than the remaining reactions of the network, assuming the network is asymptotically stable around a certain equilibrium.

Observe that for the model reduction of a given asymptotically stable biochemical reaction network governed by enzyme kinetics, it is important to determine which of the complexes are to be deleted so that the reduced model approximates the full model reasonably well.
 
\section{Yeast Glycolysis Model}
We have applied our model reduction procedure on the model of yeast glycolysis described in \cite{Karen}. The schematic of the model is shown in Fig. \ref{fig:YGM}. The corresponding detailed mathematical model can be found in \cite{Karen}.
\begin{figure}[t]
\begin{center}
\scalebox{0.9}{
\includegraphics[width=\columnwidth]{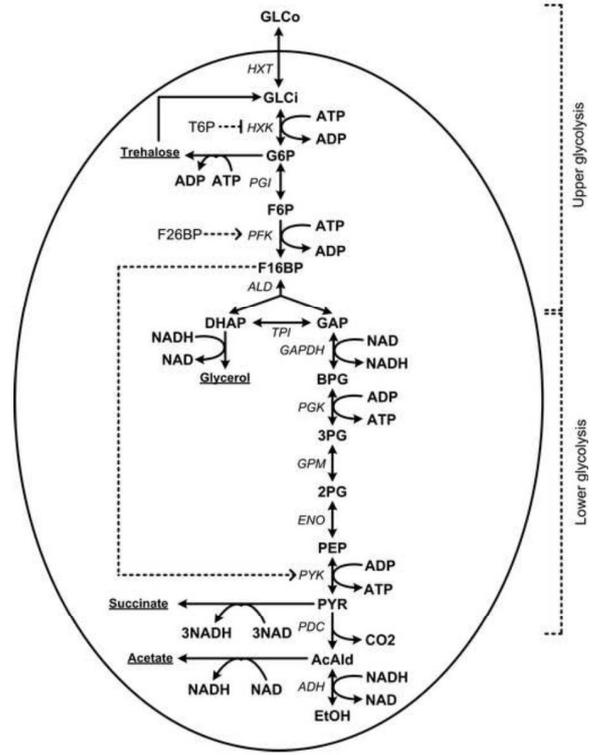}}
\end{center}
\caption{The yeast glycolysis model as discussed in \cite{Karen}. The figure is taken from \cite{Karen}.}
\label{fig:YGM}
\end{figure}

For the modelling, we have considered the following as external fluxes: 
\begin{itemize}
\item in/efflux of extracellular glucose (Glco) to/from intracellular glucose (Glci);
\item constant influx of trehalose to Glci;
\item constant efflux of trehalose from glucose 6-phosphate (G6P);
\item constant efflux of glycerol from dihydroxyacetone phosphate (DHAP);
\item efflux of succinate from pyruvate (PYR) proportional to the concentration of PYR;
\item efflux of acetate from acetaldehyde (AcAld) proportional to the concentration of AcAld;
\item in/efflux of ethanol to/from AcAld catalyzed by alcohol dehydrogenase (ADH).
\end{itemize}

The reactions of the network are governed by enzyme kinetics and we can write the equations of the model in the same form as (\ref{eq:stanform}). Conditions of a glucose pulse as described in \cite{Karen} are assumed for the model. According to these conditions, for $t<0$, concentrations of Glco and ATP are 0.2 mM and 5 mM respectively, and for $t \geq 0$, these are equal to 5 mM and 2.5 mM respectively. It is assumed that ethanol (EtoH) has a constant concentration and the network is at equilibrium for $t<0$. With these conditions, it is found that the model is asymptotically stable as $t \rightarrow \infty$.

Various combinations of complexes have been considered for deletion in order to obtain a reduced model that closely mimics the transient behaviour of the original model. It is found that deletion of the complexes G6P, 3-phosphoglycerate (3PG), 2-phosphoglycerate (2PG) and phosphoenoylpyruvate (PEP) from the original model produces the best results in this respect. The schematic of the reduced model is shown in Fig. \ref{fig:RYGM}.

\begin{figure}[t]
\begin{center}
\scalebox{0.95}{
\includegraphics[width=\columnwidth]{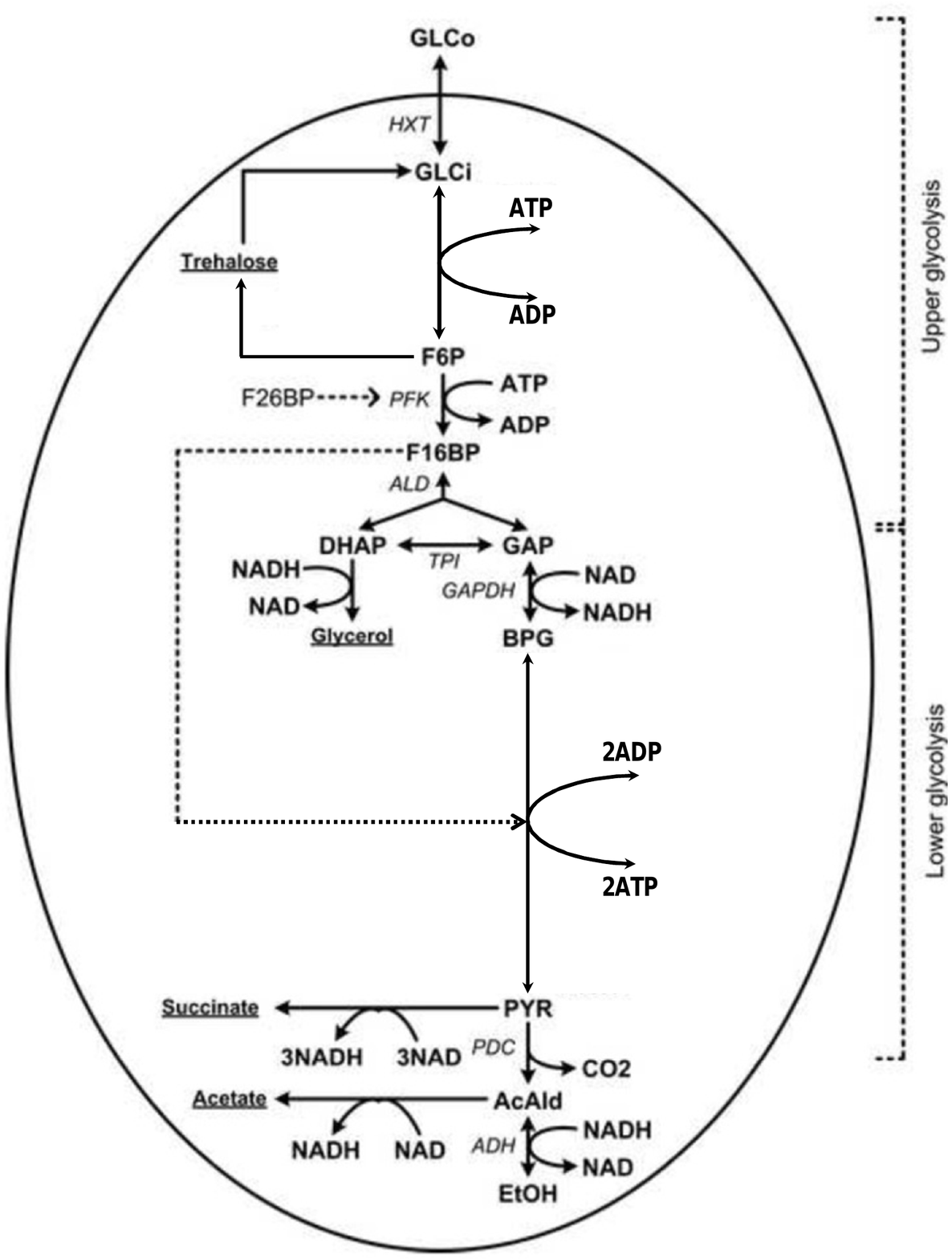}}
\end{center}
\caption{The reduced model of yeast glycolysis. The figure is a modified form of the original glycolysis model in \cite{Karen}.}
\label{fig:RYGM}
\end{figure} 

It is observed that there is a good agreement between the transient behaviours of most of the metabolites when comparing the original model to the reduced model. This implies that in the original model, under the given conditions of the glucose pulse, the metabolites G6P, 3PG, 2PG and PEP reach their steady states much faster than the remaining metabolites of the network. Thus imposing that these metabolites stay at a constant concentration has very little effect on the dynamics of the network. The graphs of comparison of the evolution of concentrations of Glci and PYR are shown in Figures \ref{fig:Glci} and \ref{fig:Pyr} respectively.

\begin{figure}[t]
\begin{center}
\scalebox{1}{
\includegraphics[width=\columnwidth]{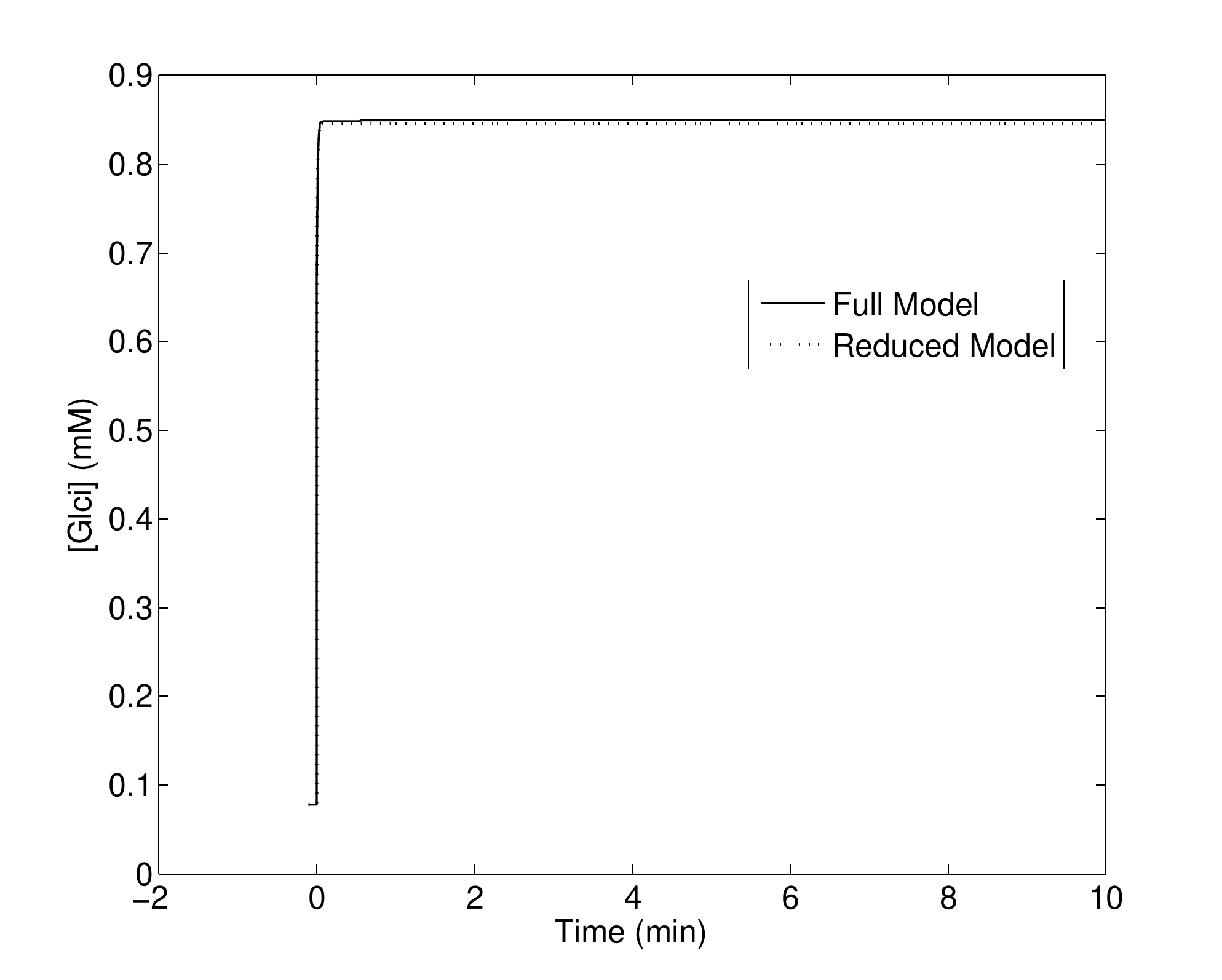}}
\end{center}
\caption{Concentration of Glci vs time}
\label{fig:Glci}
\end{figure} 
\begin{figure}[t]
\begin{center}
\scalebox{1}{
\includegraphics[width=\columnwidth]{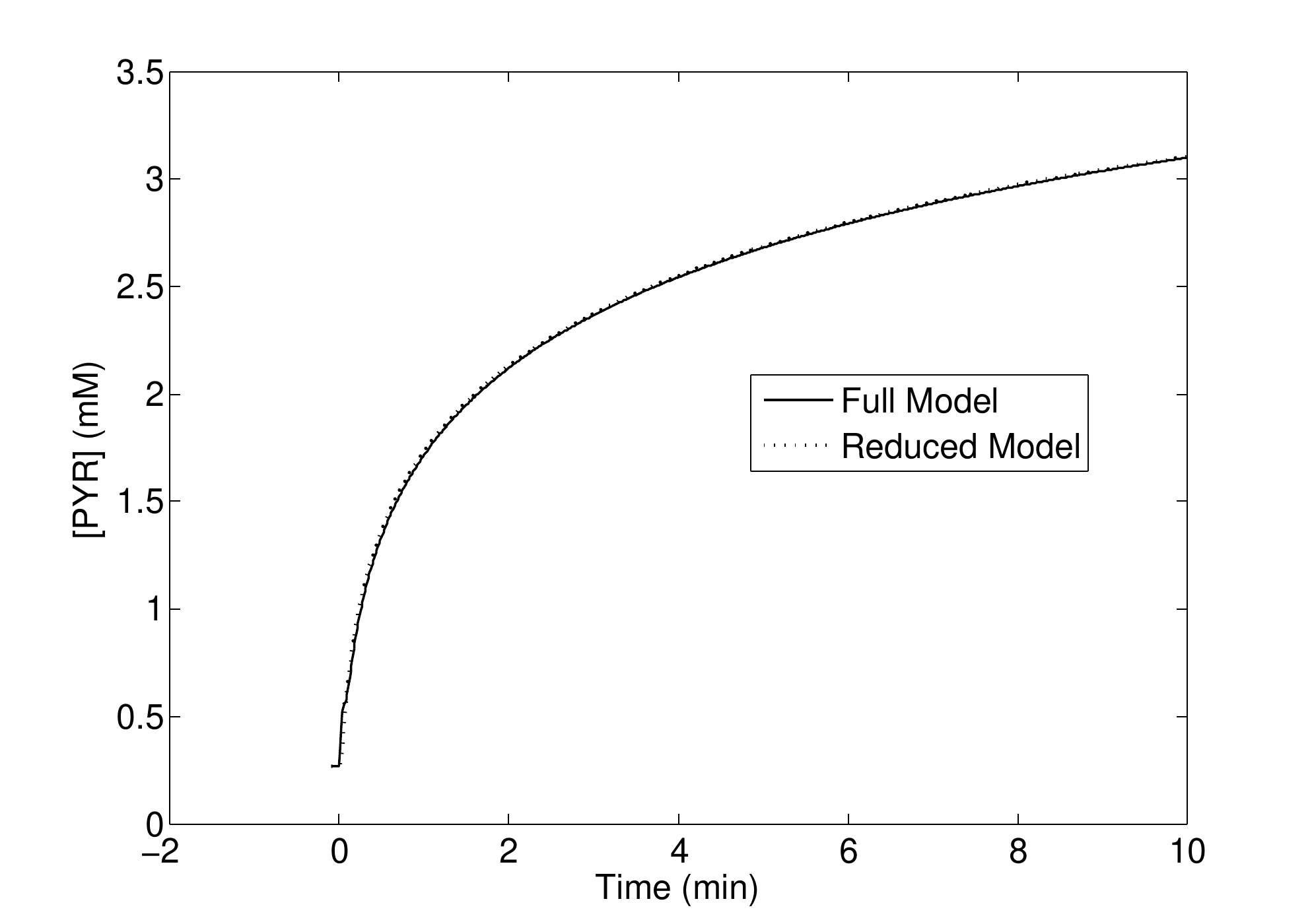}}
\end{center}
\caption{Concentration of PYR vs time}
\label{fig:Pyr}
\end{figure}

\section{Conclusion}
In this paper, we have outlined a method for model reduction of biochemical reaction networks that are asymptotically stable around an equilibrium. When perturbed from the equilibrium, such a reaction network has certain species reaching their equilibrium much faster than the remaining ones. The principle behind our model reduction method is to impose the condition that a few of such species remain at constant concentrations. This is achieved by deletion of complexes entirely made up of such species from the complex graph of the network. Apart from a reduction in the number of state variables, our model reduction method also leads to a reduction in the number of reactions and parameters of the model. We have applied our method on a yeast glycolysis model and have observed a good agreement between the transient behaviours of the original and the reduced models.  

\end{document}